\documentclass[journal]{IEEEtran}
\usepackage{graphicx,color}
\usepackage{cite}
\usepackage{setspace} 
\usepackage{amsmath}
\usepackage{amsmath,amsthm,amssymb}
\usepackage{multirow}
\usepackage{hhline}
\usepackage{epsfig}
\usepackage{subfigure}
\usepackage{epstopdf}
\usepackage{verbatim}
\usepackage{algorithm}
\usepackage{algorithmicx}
\usepackage{algpseudocode}
\usepackage{etoolbox}
\usepackage{cases}
\usepackage{mathtools,stmaryrd}
\usepackage{bbm}
\usepackage{multicol}
\usepackage{bbding}

\SetSymbolFont{stmry}{bold}{U}{stmry}{m}{n}

\makeatother

\newtheorem{remark}{Remark}
\newtheorem{theorem}{Theorem}
\newtheorem{prop}{Proposition}
\newtheorem{coro}{Corollary}

\allowdisplaybreaks

\begin{document}

\title{Differential Chaos Shift Keying-based \\Wireless Power Transfer with Nonlinearities}

\author{Authors}
\author{Priyadarshi Mukherjee, \textit{Member, IEEE}, Constantinos Psomas, \textit{Senior Member, IEEE}, and~Ioannis Krikidis, \textit{Fellow, IEEE}
\thanks{P. Mukherjee, C. Psomas, and I. Krikidis are with the Department of Electrical and Computer Engineering, University of Cyprus, Nicosia 1678 (E-mail: \{mukherjee.priyadarshi, psomas, krikidis\}@ucy.ac.cy).  Preliminary results of this work have been presented at the IEEE International Conference on Acoustics, Speech, and Signal Processing $2021$, Toronto, Canada \cite{icassp}.

This work was co-funded by the European Regional Development Fund and the Republic of Cyprus through the Research and Innovation Foundation, under the projects EXCELLENCE/0918/0377 (PRIME), INFRASTRUCTURES/1216/0017 (IRIDA) and POST-DOC/0916/0256 (IMPULSE). It has also received funding from the European Research Council (ERC) under the European Union's Horizon 2020 research and innovation programme (Grant agreement No. 819819).}}

\maketitle

\begin{abstract}
In this paper, we investigate conventional communication-based chaotic waveforms in the context of wireless power transfer (WPT). Particularly, we present a  differential chaos shift keying (DCSK)-based WPT architecture, that employs an analog correlator at the receiver, in order to boost the energy harvesting (EH) performance. We take into account the nonlinearities of the EH process and derive closed-form analytical expressions for the harvested direct current (DC) under a generalized Nakagami-$m$ block fading model. We show that, in this framework, both the peak-to-average-power-ratio of the received signal and the harvested DC, depend on the parameters of the transmitted waveform. Furthermore, we investigate the case of deterministic unmodulated chaotic waveforms and demonstrate that, in the absence of a correlator, modulation does not affect the achieved harvested DC. On the other hand, it is shown that for scenarios with a correlator-aided receiver, DCSK significantly outperforms the unmodulated case. Based on this observation, we propose a novel DCSK-based signal design, which further enhances the WPT capability of the proposed architecture; corresponding analytical expressions for the harvested DC are also derived. Our results demonstrate that the proposed architecture and the associated signal design, can achieve significant EH gains in DCSK-based WPT systems. Furthermore, we also show that, even by taking into account the nonlinearities at the transmitter amplifier, the proposed chaotic waveform performs significantly better in terms of EH, when compared with the existing multisine signals.
\end{abstract}


\IEEEpeerreviewmaketitle

\section{Introduction}

The wireless traffic has been growing at an explosive rate in recent years; according to Ericsson, it is expected to increase more than five times between $2019$ and $2025$ \cite{ericsson}. For applications like the Internet of Things and  massive machine-type communications, where a large number of devices are deployed, the overall network lifetime is  often affected due to limited battery constraints. Thus, powering or charging these devices becomes critical as well as costly. As a result, low-powered and self-sustainable next generation wireless communication networks is an important and relevant topic of research, for both industry and academia. In this context, based on the advances made in recent years, wireless power transfer (WPT) can be considered as a suitable candidate, where the devices are wirelessly powered by harvesting energy from ambient/dedicated radio-frequency (RF) signals. This is achieved by employing a rectifying antenna (rectenna) at the receiver that converts the received RF signals to direct current (DC) \cite{srv1}. One of the main advantages of this process is that it can be controlled, unlike conventional energy sources, where the available power for harvesting is itself erratic in nature \cite{srv1}.

The design of efficient WPT architectures fundamentally relies on accurate mathematical models of the harvesting circuit. While some works propose simplified linear energy harvesting (EH) models \cite{sir}, the work in \cite{plinear} proposes a piece-wise linear model based on the sensitivity and saturation levels of the circuit. On the other hand, several nonlinear EH models have also been proposed. The authors in \cite{satm} propose a tractable logistic saturation nonlinear model that originates from the saturation of the output power beyond a certain input RF power due to the diode breakdown. As this model is obtained by fitting measurements from practical RF-based EH circuits for a given excitation signal, it is a significantly improved version of its over-simplified piece-wise linear counterpart. However, all these models fail to characterize the actual working principle of the harvesting circuit. As a result, unlike the linear, piece-wise linear, and the logistic saturation nonlinear model, the work in \cite{harv} proposes a realistic circuit-based nonlinear model of the harvester circuit, which relies on the EH circuit characteristics and also enables the design of waveforms that maximize the WPT efficiency.

This model has triggered recent interests in the area of waveform design for WPT, with an objective of making the best use of the harvester to deliver a maximum amount of DC power. The authors in \cite{papr} show that the nonlinearity of the rectification process causes certain waveforms, with high peak-to-average-power-ratio (PAPR), to provide a higher DC output compared to conventional constant-envelop sinusoidal signals. Based on this observation, some works, e.g. \cite{wdesg,hparam,fhelps,npsk,bweff}, investigate the effect of transmitted waveforms and modulations on WPT. The authors in \cite{wdesg} propose the use of multisine waveforms for WPT due to their high PAPR. The work in \cite{hparam} proposes a novel simultaneous wireless information and power transfer (SWIPT) architecture based on the superposition of multi-carrier unmodulated and modulated waveforms at the transmitter. The authors in \cite{fhelps} provide insights on how fading and diversity are beneficial for boosting the RF-to-DC conversion efficiency; they develop a new  form of signal design for WPT, which relies on multiple dumb antennas at the transmitter to induce fast fluctuations of the channel. The work in \cite{npsk} proposes an asymmetric modulation scheme specifically for SWIPT, that significantly enhances the rate-energy region as compared to its existing symmetric counterpart. Apart from the multisine waveforms, experimental studies demonstrate
that due to their high PAPR, chaotic waveforms outperform conventional single-tone signals in terms of WPT efficiency \cite{chaosexp2}.

Chaotic signals have been used in the past decades for wireless information security and privacy \cite{sec1,sec2}. However, due to their properties such as aperiodicity and sensitivity to initial data, chaotic waveforms have been extensively used to improve the performance of wireless digital communication systems \cite{ch2}. In this context, differential chaos shift keying (DCSK), is one of the most widely studied chaotic signal-based communication systems \cite{ch1}. DCSK is a prominent benchmark in the class of non-coherent transmitted reference modulation techniques, which comprises of a reference and an identical/inverted replica of the reference depending on the data transmitted. The majority of the related works focus on the error performance of DCSK-based systems for various scenarios \cite{ch3,ch4,ch5}. The authors in \cite{ch3} propose an $M$-ary DCSK system, in which successive bits are converted to a symbol and then transmitted by using the same modulation scheme. The work in \cite{ch4} investigates a general set-up for evaluating the performance of DCSK system over various communication channels and the performance of a  cooperative diversity-aided DCSK-based system is analysed in \cite{ch5}.

To exploit the benefits of both DCSK and WPT, there are a some works \cite{chaoswipt1,chaoswipt2,chaoswipt3,chaoswipt4} in the literature that investigate SWIPT in a chaotic framework. In \cite{chaoswipt1}, a non-coherent short-reference DCSK SWIPT architecture is proposed by using the time switching architecture to achieve higher data rate than conventional systems. The authors in \cite{chaoswipt2} investigate SWIPT with a power-splitting receiver for a multi-carrier (MC) chaotic framework and propose two WPT protocols. In the first protocol, a fraction of the reference sub-carrier power is utilized for EH and in the second protocol, a fraction of the total sub-carrier power is used for EH. Based on the proposed MC architecture, a chaotic carrier-index (CI)  system is investigated for a basic SWIPT set-up to further reduce the energy consumption \cite{chaoswipt3}. In particular, based on the transmission characteristics of index modulation, the proposed SWIPT scheme exploits the inactive carriers of CI-DCSK to deliver energy by transmitting random noise-like signals. In \cite{chaoswipt4}, an adaptive link selection for buffer-aided relaying is investigated in a decode-and-forward relay-based DCSK-SWIPT architecture. Furthermore, by taking into account the decoding cost at the relays, two link-selection schemes based on the harvested energy, and not channel state information (CSI), are proposed.

The above studies focus on the performance of SWIPT in a DCSK-based scenario. Furthermore, they consider a simplified linear model for harvesting, which is impractical. However, the specific gains in EH from chaotic signals have not been explored. Motivated by this, in this paper, we investigate chaotic signal-based waveform designs for WPT. To the best of our knowledge, this is the first work that presents a complete analytical framework of DCSK-based WPT, by also taking into account the nonlinearities of the EH process. Specifically, the contribution of this paper is threefold.
\begin{itemize}
\item We propose a DCSK-based WPT architecture, where an analog correlator-aided EH circuit is employed at the receiver. Although WPT is the main focus of this work, we consider DCSK, i.e. an information-based waveform. We analytically derive the PAPR of the signal at the harvester as a function of the transmitted waveform parameters. We demonstrate that the analog correlator prior to the harvester allows us to control the PAPR of the signal at the input of the harvester and therefore the EH performance.

\item Analytical expressions of the harvested DC are derived for both cases, i.e. with and without the analog correlator at the receiver, under a Nakagami-$m$ block fading scenario. The derived closed-form expressions are a function of the spreading factor and the fading parameter $m$ and are verified by extensive Monte Carlo simulations. They provide a quick and convenient methodology of evaluating the system's performance and obtaining insights into how key system parameters affect the performance. In addition, our results show that similar to \cite{fhelps}, a Rayleigh fading scenario results in a higher harvested DC performance, compared to a no-fading scenario, which shows that fading is beneficial for DCSK-based WPT.

\item We apply the proposed WPT architecture to unmodulated chaotic waveforms and investigate how it performs against its modulated DCSK counterpart. We demonstrate that modulation does not affect the EH performance without the correlator at the receiver. On the contrary, with a correlator-aided receiver, DCSK significantly outperforms the unmodulated case. Based on this observation, we propose a novel short reference DCSK-based signal design, which results in further EH performance gains. Finally, we compare the proposed waveform with the existing $N$-tone multisine signals, by taking into account the imperfections of the high power amplifier (HPA) at the transmitter.
\end{itemize}
The rest of this paper is organized as follows: Section II introduces the proposed system architecture. Section III presents the complete analytical framework of chaotic signal-based WPT by considering a nonlinear rectenna model. Section IV proposes a novel WPT-optimal chaotic waveform. Numerical results are presented in Section V, followed by our conclusions in Section VI.

\section{A Chaotic Signal-based WPT System Architecture} \label{SM}

\begin{figure*}[!t]
\centering\includegraphics[width=0.8\linewidth]{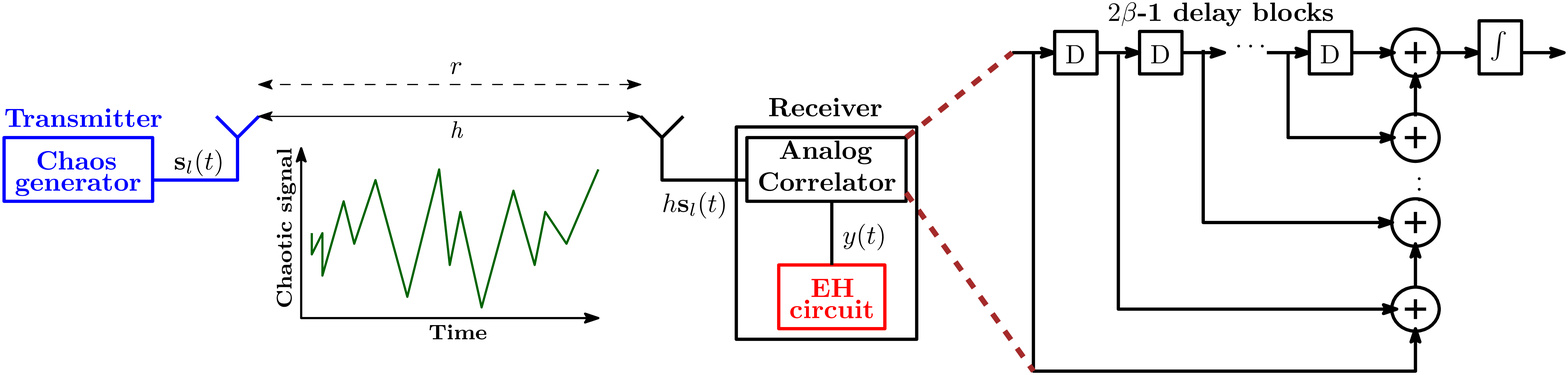}
\caption{Proposed architecture for chaotic signal-based WPT; ``D'' corresponds to the delay blocks and the symbol ``$\int$'' corresponds to the integrator.}
\label{fig:model}
\end{figure*}


\subsection{System model}

We consider a point-to-point WPT set-up, where the transmitter employs a DCSK generator and the receiver consists of an analog correlator (discussed in Section \ref{AC}) followed by an EH circuit, as depicted in Fig. \ref{fig:model}. Note that even though DCSK is primarily a communication-based signal, in this work, we focus on its effects on WPT\footnote{This assumption refers to communication networks that operate based on DCSK modulation and EH receivers are added to operate by the ambient radiation.}.

We assume that the wireless link suffers from both large-scale path-loss effects and small-scale block fading. Specifically, the received power is proportional to $r^{-\alpha}$, where $r$ is the transmitter-receiver (Tx-Rx) distance and $\alpha$ denotes the path-loss exponent. Moreover, a Nakagami-$m$ distributed block fading channel $h$ is considered with unit mean power, where $m \geq 1$.

\subsection{Chaotic signals}

Assume a DCSK signal, where the current symbol is dependent on the previous symbol [17] and different sets of chaotic sequences can be generated by using different initial conditions. Each transmitted bit is represented by two sets of chaotic signal samples, with the first set representing the reference, and the other conveying information. If $+1$ is to be transmitted, the data sample will be identical to the reference sample. Otherwise, an inverted version of the reference sample will be used as the data sample [16]. During the $l$-th transmission interval, the output of the transmitter is
\begin{align}  \label{sym}
s_{l,k}=\begin{cases} 
x_{l,k}, & k=2(l-1)\beta+1,\dots,(2l-1)\beta,\\
d_lx_{l,k-\beta}, & k=(2l-1)\beta+1,\dots,2l\beta,
\end{cases}&
\end{align}
where $d_l=\pm1$ is the information bit, $x_{l,k}$ is the chaotic sequence used as the reference signal, and $x_{l,k-\beta}$ is its delayed version. Let $\beta$ be a non-negative integer, defined as the \textit{spreading factor}. Then, $2\beta$ chaotic samples are used to spread each information bit. Furthermore, $x_{l,k}$ can be generated according to various existing chaotic maps. Due to its good correlation properties, we consider the Chebyshev map of degree $\xi$ for chaotic signal generation, which is defined as \cite{ch2}
\begin{align}
x_{k+1}=\cos(\xi\cos^{-1}(x_k)), \forall \:\:|x_k| \leq 1.
\end{align}

\subsection{Analog correlator}   \label{AC}

To further boost the EH performance, we propose a WPT receiver architecture, where an analog correlator \cite{anaco2} preceeds the EH rectifier circuit (see Fig. \ref{fig:model}). The motivation behind the employment of an analog correlator\footnote{It is to be noted that the analog correlator is an active device, but with a nominal power consumption. Specifically, it has been observed, that a CMOS-implemented 11-bit analog correlator consumes 40 nW of power \cite{anaco2}. As a result, we have not considered this energy consumption in the context of applying the correlator.} block is two fold: i) it enables controlling the signal at the input of the harvester, and ii) it enhances the PAPR of the signal. An analog correlator primarily consists of a series of delay blocks, which result in signal integration over a specified period of time; an ideal $\psi$-bit analog correlator consists of $(\psi-1)$ number of delay blocks \cite{anaco2}.
In what follows, for the sake of simplicity, we consider $\psi$ as the transmitted DCSK symbol length,  i.e. $\psi=2\beta$. Accordingly, the analog correlator output $y_l(t)$ for the $l$-th transmitted symbol is
\begin{equation}    \label{corr}
y_l(t)=\sqrt{P_t}h_l\sum\limits_{k=1}^{2\beta} s_{l,k}(t),
\end{equation}
where $P_t$ is the transmission power. Note that we propose the use of a correlator and not a summing circuit in Rx. A summing circuit generates an output, which is the sum of multiple input signals \cite{ssmith} whereas the analog correlator allows an effective integration of the same signal over a certain time interval \cite{anaco2}, i.e. it generates an output, which is the sum of the delayed versions of the same signal. In other words, the input signal is correlated with a delayed version of itself. Note that the value $\psi=1$ corresponds to the conventional case without a correlator.

In the following proposition, we show the effect of the analog correlator on the signal's PAPR.
\begin{prop}    \label{prp1}
The signal PAPR at the harvester input is
\begin{align}  \label{prop1}
\mathrm{PAPR}=\begin{cases} 
2, & \text{without correlator} \:\: (\psi=1),\\
4\beta, & \text{with correlator} \:\: (\psi=2\beta).
\end{cases}&
\end{align}
\end{prop}

\begin{proof}
See Appendix \ref{app0}.
\end{proof}

\noindent Therefore, since high PAPR signals are desirable for WPT \cite{papr}, the correlator significantly enhances the WPT performance.

\subsection{Energy transfer model} \label{ehsec}

The WPT receiver is equipped with an antenna followed by a rectifier. The rectifier, which generally consists of a diode  (e.g., a Schottky diode) and a passive low pass filter, acts as an envelope detector \cite{envd} and therefore neglects the phase of the received signal $y(t)$. Based on the nonlinearity of this circuit, the output DC current is approximated in terms of $y(t)$ as \cite{harv}
\begin{equation} \label{brunoeh}
z_{\rm DC}=k_2R_{ant}\mathbb{E} \{ |y(t)|^2 \}+k_4R_{ant}^2\mathbb{E} \{ |y(t)|^4 \},
\end{equation}
where the parameters $k_2,k_4,$ and $R_{ant}$ are constants determined by the characteristics of the circuit. Note that the conventional linear model is a special case of this nonlinear model and can be obtained by considering only the first term in \eqref{brunoeh}. For the sake of presentation, we will use $\varepsilon_1 = r^{-\alpha}k_2R_{ant}P_t$ and $\varepsilon_2=r^{-2\alpha}k_4R_{ant}^2P_t^2$. Hence, from \eqref{corr} and \eqref{brunoeh}, we obtain
\begin{align} \label{zdef1}
z_{\rm MC}&=\varepsilon_1\mathbb{E}\left\lbrace \left(  |h|\sum\limits_{k=1}^{2\beta}s_k \right) ^2 \right\rbrace  + \varepsilon_2\mathbb{E}\left\lbrace \left(  |h|\sum\limits_{k=1}^{2\beta}s_k \right) ^4 \right\rbrace,
\end{align}
for the scenario with an analog correlator at the receiver, where the expectation is taken over both $|h|$ and $s_k$. On the other hand, the output DC without a correlator is
\begin{align} \label{zdef2}
z_{\rm MNC}&=\varepsilon_1\mathbb{E}\left\lbrace |h|^2\sum\limits_{k=1}^{2\beta}s_k^2 \right\rbrace  + \varepsilon_2\mathbb{E}\left\lbrace |h|^4\sum\limits_{k=1}^{2\beta}s_k^4 \right\rbrace,
\end{align}
where $2\beta$ samples are considered for fair comparison. This implies that $\beta$ will have a considerable impact on the harvested energy. Finally, note that the subscript MC and MNC refers to the scenario of modulated waveform transmission (DCSK in this case), with and without the correlator at the receiver, respectively.

\section{Chaotic Signal-based Wireless\\ Power Transfer} \label{harvsec}

In this section, we investigate the effect of DCSK signals on WPT. We analyze the performance of the proposed receiver architecture in terms of the harvested energy, when DCSK waveforms are transmitted. Towards this direction, we provide the following theorem.

\begin{theorem} \label{nakath1}
For a WPT receiver with a correlator, the harvested DC is given by
\begin{equation}  \label{theocnaka}
z_{\rm MC}=\varepsilon_1\beta+\varepsilon_2\frac{3(1+m)}{m}\beta(2\beta-1).
\end{equation}
\end{theorem}
\begin{proof}
See Appendix \ref{napp1}.
\end{proof}

\noindent Theorem \ref{nakath1} provides a generalized closed-form expression for $z_{\rm MC}$ in terms of $m$. Note that $z_{\rm MC}$ corresponding to a no-fading scenario ($m \to \infty$) can also be obtained as a special case, given in the following corollary.

\begin{coro}
In a no-fading scenario, $z_{\rm MC}$ is given by
\begin{equation}
\lim\limits_{m \rightarrow \infty} z_{\rm MC}=\varepsilon_1\beta+3\varepsilon_2\beta(2\beta-1).
\end{equation}
\end{coro}

The above corollary follows directly from Theorem \ref{nakath1}, as $(1+m)/m \to 1$ for $m\to\infty$.

\begin{remark}
Note that in Theorem \ref{nakath1}, $m=1$ results in an enhanced $z_{\rm MC}$ compared to the $m \rightarrow \infty$ scenario, i.e. wireless fading enhances WPT. This observation corroborates the claims made in \cite{fhelps} regarding the beneficial role of fading in WPT systems.
\end{remark}

Next, we consider the case of a conventional WPT receiver without the analog correlator, which is given by the following theorem.
\begin{theorem} \label{nakath2}
Without a correlator at the WPT receiver, the harvested DC is
\begin{equation}  \label{theonaka}
z_{\rm MNC}=\varepsilon_1\beta+\varepsilon_2\frac{3(1+m)}{4m}\beta.
\end{equation}
\end{theorem}

\begin{proof}
See Appendix \ref{napp2}.
\end{proof}

\begin{remark}
From Theorems \ref{nakath1} and \ref{nakath2}, we observe that $z_{\rm MC}$ and $z_{\rm MNC}$ is a quadratic and linear function of $\beta$, respectively. This highlights the benefits of the proposed architecture based on the nonlinearity of the EH process. Note that, if a linear EH model is considered, which only accounts for the second-order term in \eqref{brunoeh}, we obtain $z_{\rm MC}=z_{\rm MNC}$.
\end{remark}

We observe that, both $z_{\rm MC}$ and $z_{\rm MNC}$ are inversely proportional to $m$. We also note that the analytical expressions obtained can be extended to Rician fading scenarios (as special cases), by replacing $m=\frac{(K+1)^2}{2K+1}$, where $K$ is the Rice factor \cite{rnaka}.

For the sake of completeness, we also provide $z_{\rm MNC}$ corresponding to a no-fading scenario.

\begin{coro}
The harvested DC $z_{\rm MNC}$ corresponding to a no-fading scenario is given by
\begin{equation} \label{propb1n}
\lim\limits_{m \rightarrow \infty} z_{\rm MNC}=\varepsilon_1\beta+\frac{3}{4}\varepsilon_2\beta.
\end{equation}
\end{coro}

\section{DCSK-based Waveform Design for WPT}  \label{info}

We now extend the analytical model, initially proposed for DCSK signals, to unmodulated (deterministic) chaotic waveforms. Then, we modify the conventional DCSK waveform at the transmitter, in a way that further enhances the WPT performance.

\begin{figure*}[!t]
\centering\includegraphics[width=0.76\linewidth]{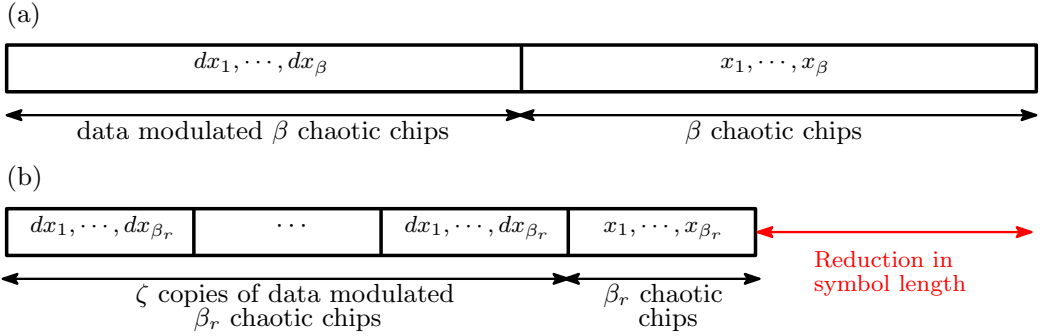}
\caption{(a) DCSK frame, and (b) SR-DCSK frame.}
\label{fig:sr}
\end{figure*}

\subsection{Unmodulated chaotic signals}

We consider the case of using a deterministic, unmodulated chaotic waveform for WPT and compare its performance with the conventional DCSK signal. For this case, the $k$-th bit of the $l$-th transmitted unmodulated chaotic symbol is given by
\begin{equation} \label{corr3}
s_{l,k}=x_{l,k}, \quad k=2(l-1)\beta+1,\dots,2l\beta.
\end{equation}
It is worth noting that while it is essential to have a symbol length of $2\beta$ in DCSK, such a constraint is not required in the unmodulated chaotic case. However, we consider a symbol length $2\beta$ for a fair comparison of the EH performance. Now, we state the following proposition, where we obtain the harvested DC for unmodulated chaotic signals, with/without a correlator at the WPT receiver.
\begin{prop} \label{prop2}
Based on the symbol structure as defined in \eqref{corr3}, we obtain
\begin{equation}    \label{corrnaka1}
z_{\rm UM,C}=\varepsilon_1\beta+\varepsilon_2\frac{3(1+m)}{m}\beta\left(\beta-\frac{1}{4}\right),
\end{equation}
\textit{with a correlator and,}
\begin{equation}    \label{corrnaka2}
z_{\rm UM,NC}=\varepsilon_1\beta+\varepsilon_2\frac{3(1+m)}{4m}\beta,
\end{equation}
\textit{without a correlator.}
\end{prop}

\begin{proof}
See Appendix \ref{napp3}.
\end{proof}

\begin{remark} \label{rem}
Without a correlator at the receiver, we have $z_{\rm MNC}=z_{\rm UM,NC}$, for a given set of system parameters (Theorem \ref{nakath2} and Proposition \ref{prop2}). On the other hand, when a correlator is employed, we have $z_{\rm MC} > z_{\rm UM,C}$, with a performance gain equal to $\varepsilon_2\frac{3(1+m)}{m}\beta(\beta-\frac{3}{4})$ (Theorem \ref{nakath1} and Proposition \ref{prop2}). In other words, a modulated DCSK signal positively affects the EH performance in the presence of a correlator.
\end{remark}

The above remark can be explained as follows. In contrast to the unmodulated waveform, DCSK also carries information due to its inherent randomness, i.e. the first $\beta$ symbols are chaotic, followed by their replica or inverted replica. In case of a correlator-aided receiver, this randomness of the DCSK waveform induces fluctuations of the transmit signal amplitude, which boosts the harvested DC. The difference in performance, given in Remark \ref{rem}, is in line with the claims made in \cite{hparam} regarding the importance of using modulated multisine waveforms for efficient WPT.

Note that Remark \ref{rem} considers identical system parameters and similar channel conditions between the two cases. However, this claim does not hold for all possible scenarios. Indeed, an interesting insight can be obtained when we analyze the energy harvested in the modulated case under the deterministic fading scenario. Specifically, let $m_1$ and $m_2$ be the fading parameters of the modulated and unmodulated case, respectively. Then, the performance gap $\Delta = z_{\rm MC}-z_{\rm UM,C}$ can be written as
\begin{equation}  \label{gap}
\Delta = \varepsilon_2\frac{3(1+m_1)}{m_1}\beta(2\beta-1)-\varepsilon_2\frac{3(1+m_2)}{m_2}\beta\left(\beta-\frac{1}{4}\right).
\end{equation}
Observe that for the special case of $m_1=m_2$, we always have $\Delta >0$, i.e. $z_{\rm MC}>z_{\rm UM,C}$ (Remark \ref{rem}). On the other hand, $z_{\rm MC}$ in a no-fading environment ($m_1 \rightarrow \infty$) and $z_{\rm UM,C}$ in a Rayleigh fading scenario ($m_2=1$) results in $\Delta=-\frac{3}{2}\varepsilon_2\beta<0$.
This implies that, with a correlator-aided receiver, an unmodulated waveform in a Rayleigh fading environment results in more harvested energy as compared to its modulated counterpart in a no-fading environment.

Furthermore, we also observe from \eqref{gap} that for finite $m_1$ and $m_2$, $\Delta >0$ holds if and only if
\begin{align}  \label{ineq}
\frac{(1+m_1)}{m_1}(2\beta-1)-\frac{(1+m_2)}{m_2}\left(\beta-\frac{1}{4}\right)>0.
\end{align}
Note that $m_1$ and $m_2$ are necessarily channel parameters, i.e. they cannot be controlled. However, it is possible to tune $\beta$ accordingly to satisfy \eqref{ineq}. After some trivial algebraic manipulations, we obtain $\beta > \beta_{\rm opt}^+$, where
\begin{equation}  \label{bopt}
\beta_{\rm opt} = \frac{\left(\frac{1+m_1}{m_1}\right)-\frac{1}{4}\left(\frac{1+m_2}{m_2}\right)}{2\left(\frac{1+m_1}{m_1}\right)-\left(\frac{1+m_2}{m_2}\right)}=\frac{1}{4}\left( \frac{4m_2+3m_1m_2-m_1}{2m_2+m_1m_2-m_1} \right)
\end{equation}
and $x^+=\max \{x,0\}$.

Thus, it is important to note that $\Delta >0$ does not hold for all values of $\beta$ and in some scenarios, we have $\Delta<0$, i.e. it is possible to obtain more harvested energy from an unmodulated waveform as compared to its modulated counterpart.

\subsection{WPT optimal DCSK-based waveform}    \label{optframe}

In Section \ref{harvsec}, we showed that a communication-based chaotic waveform with the proposed architecture boosts the WPT performance. The drawback of this approach is a long symbol duration, which also implies that a large number of chaotic chips are generated for every single symbol transmission. To this end, a shorter symbol duration may address
both these problems and accordingly, a short reference DCSK (SR-DCSK) is proposed in \cite{srdcsk}.

An unmodulated chaotic component of length $\beta_r<\beta$ is considered in SR-DCSK, followed by $\zeta$ copies of its replica, multiplied with the information such that $\beta=\zeta\beta_r$. Hence, the symbol duration is $\beta_r(1+\zeta)=\beta_r+\beta<2\beta$, i.e. the transmitted symbol is characterized by a sequence of $\beta_r+\beta$ samples of the chaotic basis signal, which is represented as
\begin{align}  \label{symsr}
s_k=\begin{cases} 
x_k, & 0 < k \leq \beta_r,\\
dx_{k-\beta_r}, & \beta_r < k \leq \beta_r+\beta,
\end{cases}&
\end{align}
\noindent where $d$ is the transmitted information bit. An illustrative comparison of DCSK and SR-DCSK is demonstrated in Fig. \ref{fig:sr}. While Fig. \ref{fig:sr}(a) shows the conventional DCSK symbol structure, the short reference SR-DCSK is depicted in Fig. \ref{fig:sr}(b). It is important to note here that the SR-DCSK is a technique proposed mainly for the purpose of efficient information transfer, such that its error performance always remains above a certain acceptable threshold. In contrast, we investigate the impact of having a shorter symbol duration on WPT and accordingly we propose a WPT optimal DCSK-based waveform.

\begin{figure*}[!t]
\begin{subnumcases}{\label{beff} z_{\rm SR}=} 
\frac{1}{2}\varepsilon_1\beta^2+\frac{3(1+m)}{8m}\varepsilon_2\beta^4, & $\beta_r=0,$ \label{beff1}\\
\varepsilon_1 \frac{\beta_r^2+\beta^2}{2\beta_r}+\varepsilon_2\frac{3(1+m)}{8m}\left( 1+6\zeta^2+\zeta^4 \right) \left( 2\beta_r^2-\beta_r \right), & $\beta_r>0.$ \label{beff2}
\end{subnumcases}
\hrulefill
\end{figure*}

Given that $\zeta$ replicas of the $\beta_r$ length reference signal are concatenated in SR-DCSK, this implies that the reference signal is partially correlated over $\zeta$ consecutive samples of the frame, i.e. the degree of `randomness' increases within a single frame, as compared to DCSK. In other words, this results in an increased correlation within the frame, which leads to enhanced WPT. This is in line with the claim made in \cite{corrl} regarding the effect of correlation in transmitted signals with respect to WPT.

Hence, we investigate SR-DCSK to obtain a WPT-optimal frame structure. As the symbol structure in \eqref{symsr} is dependent on $\beta_r$, the two extreme cases are $\beta_r=0$ and $\beta_r=\beta$. Accordingly, the WPT performance for all the possible values of $\beta_r$ is given by the following theorem.

\begin{theorem}    \label{nakath3}
A $\beta_r$-SR-DCSK results in $z_{\rm SR}$ given by \eqref{beff}.
\end{theorem}

\begin{proof}
See Appendix \ref{napp4}.
\end{proof}
\noindent It is interesting to note from \eqref{beff}, that $\beta_r$-SR-DCSK converges to DCSK for $\beta_r=\beta.$ We observe that the increased degree of `randomness' in $\beta_r$-SR-DCSK results in a significant performance enhancement in terms of WPT, with an increase of $z_{\rm SR}$ proportional to $\beta^4$, unlike $z_{\rm MC}$ in Theorem \ref{nakath1}.

\begin{prop} \label{prop3}
For a given $\beta$, the optimal $z_{\rm SR}$ is achieved at $\beta_r=1$.
\end{prop}

\begin{proof}
We know that $\beta=\zeta\beta_r$ and $\beta \in \mathbb{Z}^+$. Hence, by comparing the second and fourth order terms of the EH process separately in Theorem \ref{nakath3}, we can state that $z_{\rm SR}$ is maximum for $\beta_r=1$, with all other system parameters remaining constant.
\end{proof}

\noindent Hence, the WPT-optimal transmitted symbol of $(\beta+1)$ duration is defined as
\begin{align}  \label{optsym}
s_k=\begin{cases} 
x_k, & k=1,\\
dx_1, & 1 < k \leq \beta+1.
\end{cases}
\end{align}
The corresponding harvested DC
\begin{equation} \label{zopt}
z_{\rm SR}^{\rm opt}=\frac{1}{2}\varepsilon_1(1+\beta^2)+\frac{3(1+m)}{8m}\varepsilon_2(1+6\beta^2+\beta^4),
\end{equation}
is obtained by replacing $\beta_r=1$ in \eqref{beff2}. It is further interesting to observe that having a single chip of chaotic signal followed by $\beta$ chips of information modulated chaos in a symbol of $\beta+1$ duration results in a maximum possible correlation; this is reflected in the WPT gain of the proposed DCSK-based waveform. Therefore for large $\beta$, the reduction in symbol duration is approximately $2\beta-(\beta+1)\approx \beta$ per symbol duration; this also suggests that one chaotic chip needs to be generated per symbol compared to $\beta$ chaotic chips per symbol in DCSK. Finally, a shorter symbol duration also signifies the reduced receiver's complexity.

Now, we compare the proposed chaotic waveform with the existing multisine signal-based EH framework \cite{wdesg}. We know that for a $N$-tone multisine signal, the PAPR at the harvester input is $2N$ \cite{papr}. On the contrary, we have proved in Proposition $1$, that the PAPR of a DCSK-based signal is $4\beta$. As a high PAPR translates to a better EH performance \cite{papr}, we can intuitively say that the proposed waveform will result in a higher harvested energy, even for a set of comparable $N$ and $\beta$. Furthermore, for an $N$-tone multisine signal, the linear term of the harvested DC is independent of $N$ and the nonlinear term is linearly dependent on $N$ \cite{wdesg}. On the contrary, in the case of the proposed $\beta_r$-SR-DCSK waveform, the linear and nonlinear terms of $z_{\rm SR}^{\rm opt}$ (Eq. \eqref{zopt}) are proportional to $\beta^2$ and $\beta^4$, respectively. As a result, the proposed architecture for $\beta_r$-SR-DCSK significantly outperforms multisine waveforms, in terms of WPT.
\begin{table*} [!t] 
\begin{center}
  \caption{Summary of results.}
\resizebox{0.88\textwidth}{!}{%
{\renewcommand{\arraystretch}{2} 
  \begin{tabular}{|c||c|c|c|c|}
    \hline
    \bf{Waveform} & \bf{Modulation} & \bf{Correlator} & \bf{Harvested DC}\\
    \hline\hline
    $z_{\rm MC}$  & DCSK & $\checkmark$ & $\displaystyle\varepsilon_1\beta+\varepsilon_2\frac{3(1+m)}{m}\beta(2\beta-1)$ \\
    \hline
    $z_{\rm MNC}$  & DCSK & $\times$ & $\displaystyle\varepsilon_1\beta+\varepsilon_2\frac{3(1+m)}{4m}\beta$\\
    \hline
    $z_{\rm UM,C}$  & $\times$& $\checkmark$ & $\displaystyle\varepsilon_1\beta+\varepsilon_2\frac{3(1+m)}{m}\beta\left(\beta-\frac{1}{4}\right)$\\
    \hline
    $z_{\rm UM,NC}$  & $\times$& $\times$ & $\displaystyle\varepsilon_1\beta+\varepsilon_2\frac{3(1+m)}{4m}\beta$\\
    \hline
    $z_{\rm SR}$  & $\beta_r$-SR-DCSK & $\checkmark$ &  $\begin{aligned}[t]
&\displaystyle\frac{1}{2}\varepsilon_1\beta^2+\frac{3(1+m)}{8m}\varepsilon_2\beta^4, \qquad\qquad\qquad\qquad\qquad\qquad\quad \beta_r=0,\\
& \displaystyle\varepsilon_1 \frac{\beta_r^2+\beta^2}{2\beta_r}+\varepsilon_2\frac{3(1+m)}{8m}\left( 1+6\zeta^2+\zeta^4 \right) \left( 2\beta_r^2-\beta_r \right), \quad \beta_r>0
\end{aligned}$\\
\hline
  \end{tabular}
  }
  }
  \label{tab:summ}
\end{center}
\end{table*}
Finally, Table \ref{tab:summ}  presents a summary of the main analytical results derived in Section \ref{harvsec} and \ref{info}.

\begin{remark}
It may appear from the analytical results in Table \ref{tab:summ} that the limiting case of $\beta \rightarrow \infty$ will result in an infinite amount of harvested energy. However, note that the energy transfer model considered in this work is based on the assumption that the harvester operates in the nonlinear region \cite{wdesg}. If $\beta$ becomes too large, the diode inside the harvester will be forced into the saturation region of operation, making the derived analytical results inapplicable.
\end{remark}

\section{Numerical Results}

We validate our theoretical analysis through extensive Monte-Carlo simulations. Without any loss of generality, we consider a transmission power of $P_t=30$ dBm, a Tx-Rx distance $r=20$ m, and a pathloss exponent $\alpha=4$. The parameters considered for the EH model are taken as $k_2=0.0034,k_4=0.3829,$ and $R_{ant}=50$ $\Omega$ \cite{hparam}. Also recall that, $\psi=2\beta$ and $\psi=1$ corresponds to a WPT receiver with and without the analog correlator, respectively.

\begin{figure}[!t]
	\centering\includegraphics[width=\linewidth]{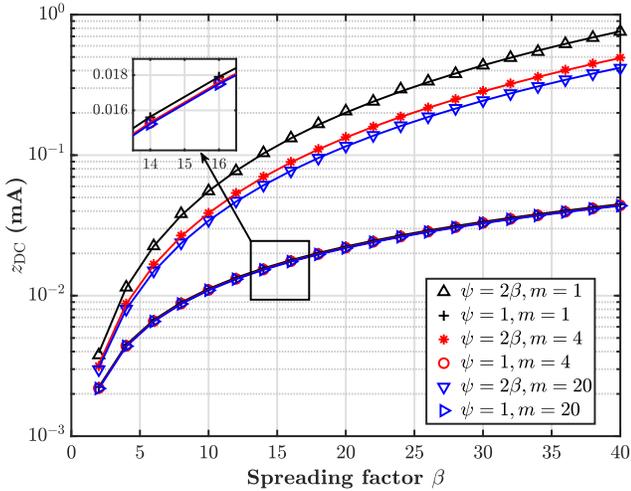}
	\caption{Effect of the spreading factor on $z_{\text{DC}}$; lines correspond to analysis and markers correspond to simulation results.}
	\label{fig:naka_comp}
\end{figure}

\subsection{Performance evaluation of the proposed architecture}\label{result1}
Fig. \ref{fig:naka_comp} depicts the WPT performance of DCSK with respect to the spreading factor, with and without the correlator at the receiver. Firstly, we observe the significant gains in WPT performance achieved with the employment of the correlator; the theoretical results (lines) match very closely with the simulation results (markers); this verifies our proposed analytical framework. This is related to the high PAPR, which is a function of the spreading factor $\beta$ as stated in Proposition \ref{prp1}. Moreover, we observe that, the harvested energy decreases with increasing $m$ for a given $\beta$. This observation is inline with \eqref{theocnaka} and \eqref{theonaka}, where we note that, the harvested DC varies with the quantity $\frac{1+m}{m}$, while the other system parameters remain constant. Specifically, due to this ratio, the best performance is achieved with $m=1$, whereas the worst is obtained with $m\to\infty$. Moreover, the loss in performance between consecutive values of $m$ decreases as $m$ increases; for example, observe the performance gap between $m=1$, $m=4$, and $m=20$.

\subsection{Effects of data modulation on WPT}
Fig. \ref{fig:unmodcomp} illustrates the impact of using a modulated waveforms for WPT and compares the performance with unmodulated chaotic signals in terms of WPT. The figure demonstrates that data modulation does not affect the WPT performance without the correlator at the receiver, i.e. both DCSK and unmodulated chaotic symbols result in an identical harvested DC. The only difference is the degradation in WPT performance with increasing $m$, which has been discussed earlier in Fig. \ref{fig:naka_comp}. On the other hand, we observe that modulating the chaotic symbols, it has a significant impact on the WPT performance. Finally, it is interesting to note that the harvested DC corresponding to Rayleigh fading, i.e. $m=1$, with unmodulated chaotic symbols is approximately identical to the harvested DC obtained with modulated symbols and $m=10$. This verifies \eqref{gap}, where in this case, the performance gap can be obtained by letting $m_1=10$ and $m_2=1$.

\begin{figure}[!t]
\centering\includegraphics[width=\linewidth]{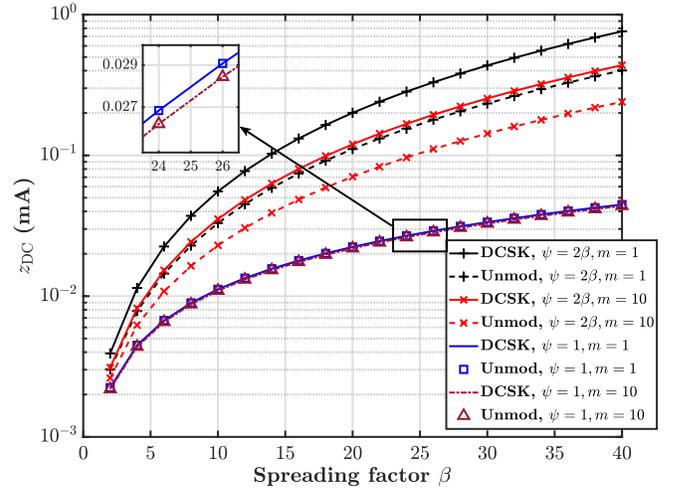}
\caption{Effect of the modulation on $z_{\text{DC}}$.}
\label{fig:unmodcomp}
\end{figure}

\subsection{Effects of symbol length on WPT}

\begin{figure}[!t]
\centering\includegraphics[width=\linewidth]{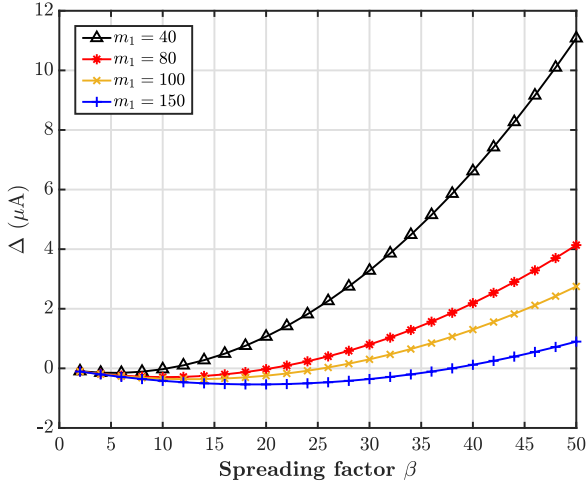}
\caption{Effect of the symbol length on $z_{\text{DC}}$.}
\label{fig:bound}
\end{figure}

Fig. \ref{fig:bound} depicts the importance of the choice of $\beta$ as a function of the channel parameter $m$. We consider the WPT performance of DCSK as a function of $\beta$, for various fading scenarios $m_1=40,80,100,150$ and compare it with its unmodulated counterpart in a Rayleigh fading scenario $(m_2=1)$. Recall that a positive $\Delta$ implies that DCSK performs better in terms of WPT and vice-versa. Observe that for all values of $m_1$, the parameter $\Delta$ attains a positive value only beyond a certain $\beta$, i.e. the unmodulated waveform performs better till that point and this value of $\beta$ increases with increasing $m_1$; this verifies our analysis to obtain \eqref{bopt}. Hence, we can rightly state that modulating the chaotic waveform does not guarantee enhanced EH. Even in the proposed architecture, i.e. with a correlator at the receiver, the choice of $\beta$ should be based on both the channel condition and the modulation, to guarantee a better EH performance.

\subsection{Effects of a shorter reference on WPT}

\begin{figure}[!t]
\centering\includegraphics[width=\linewidth]{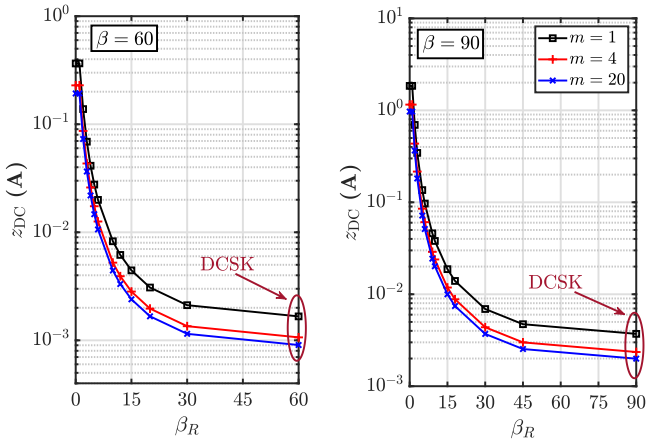}
\caption{Effect of a short reference on WPT; lines correspond to analysis and markers correspond to simulation results.}
\label{fig:zeta}
\end{figure}

Fig. \ref{fig:zeta} demonstrates the effect of an SR symbol length on the WPT performance of the proposed correlator-based architecture. The effect of having a reference length $\beta_r<\beta$, with $\zeta$ copies of its data modulated replica concatenated together such that $\beta=\zeta\beta_r$, is demonstrated here, where we consider two cases of $\beta=60$ and $\beta=90$, respectively. This figure illustrates the impact of correlation in enhancing chaotic signal-based WPT. We observe that for a given $\beta$, increasing $\beta_r$ results in obtaining the maximum harvested DC at $\beta_r=1$, followed by a monotonically decreasing performance. This is justified by the fact that as we have $\beta=\zeta\beta_r$, a higher $\beta_r$ implies a lower $\zeta$ for a given $\beta$, i.e. the correlation decreases, which results in deteriorating WPT performance. Finally, we observe that, for the special case of $\beta_r=\beta$, the WPT performance coincides with that of the DCSK-based transmission.

\subsection{Evaluation of WPT-optimal SR-DCSK}

Fig. \ref{fig:wptopt} demonstrates the performance of the proposed WPT-optimal SR-DCSK. The figure exhibits the effect of the spreading factor for two scenarios of Tx-Rx distance $20$ m and $30$ m, respectively. While the analytically obtained values match closely with that of Monte Carlo simulation, we also note that the harvested DC with Tx-Rx distance $30$ m is less compared to the harvested DC with Tx-Rx distance $20$ m; this is intuitive due to the path-loss factor.

\begin{figure}[!t]
\centering\includegraphics[width=\linewidth]{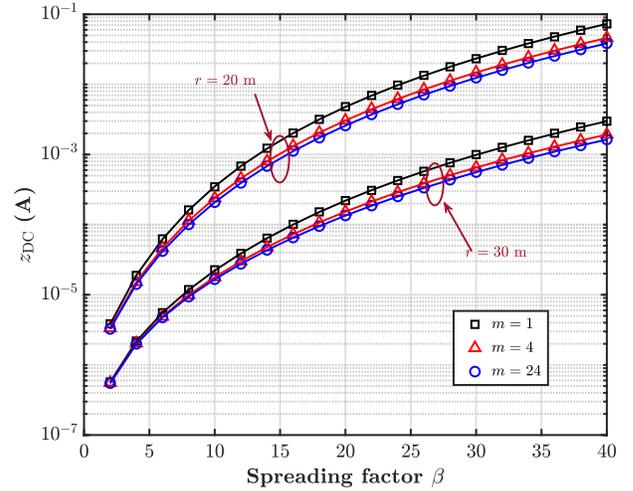}
\caption{Performance of WPT-optimal SR-DCSK; lines correspond to analysis and markers correspond to simulation results.}
\label{fig:wptopt}
\end{figure}

Fig. \ref{fig:joint} illustrates the joint effect of fading parameter $m$ and spreading factor $\beta$ on the WPT performance of the proposed symbol structure. We observe that the performance of WPT-optimal SR-DCSK improves with increasing $\beta$ and decreasing $m$ for a given Tx-Rx distance, where the harvested DC is maximum at $m=1$ and $\beta=30$, i.e. minimum $m$ and maximum $\beta$. We also note that even though the harvested energy increases monotonically with $\beta$ for a given $m$, its variation against $m$ for any particular $\beta$ is not the same. The harvested energy saturates with increasing $m$ for a given $\beta$; this observation can be justified by the fact that apart from $\beta$, it is also a function of $\frac{1+m}{m}$.

\begin{figure}[!t]
\centering\includegraphics[width=\linewidth]{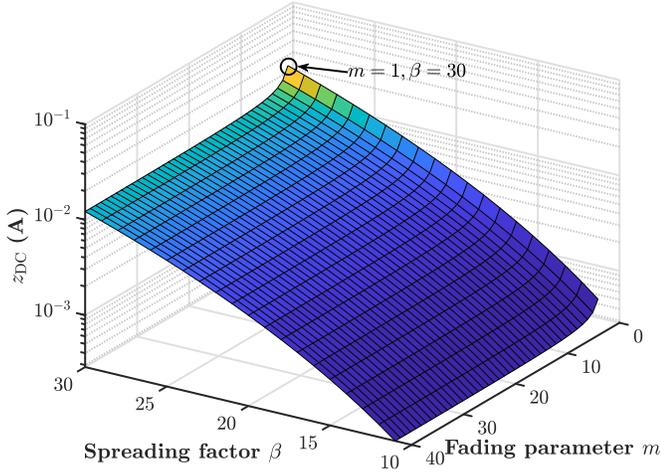}
\caption{Joint effect of the spreading factor $\beta$ and the fading parameter $m$ on WPT-optimal SR-DCSK.}
\label{fig:joint}
\end{figure}

\subsection{Comparison of WPT-optimal SR-DCSK and DCSK}

Fig. \ref{fig:dcsk_comp_sr} compares the performance of the proposed WPT-optimal SR-DCSK and DCSK against the transmit symbol length for several values of the fading parameter $m$. As expected, the proposed waveform results in a significantly higher harvested DC at approximately half the corresponding DCSK symbol length. A small symbol length implies less number of chaotic signals  generated per symbol transmission and also a smaller number of delay blocks inside the analog correlator, i.e. reduction in both receiver's complexity and cost. Besides having a short frame duration, the proposed frame structure successfully exploits the effect of correlation in the transmitted signal.

\begin{figure}[!t]
\centering\includegraphics[width=\linewidth]{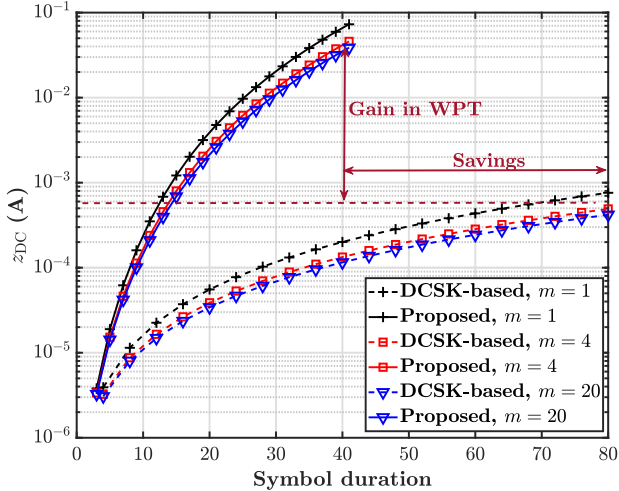}
\caption{Performance comparison of the proposed WPT-optimal SR-DCSK and DCSK.}
\label{fig:dcsk_comp_sr}
\end{figure}

\subsection{Performance comparison with multisine waveforms}

Fig. \ref{fig:mult} compares the WPT performance of the proposed chaotic and existing different $N$-tone multisine waveforms \cite{wdesg}. Furthermore, in this figure, we also consider the HPA non-linearities at the transmitter \cite{power}. The figure corroborates our claim that the proposed waveform significantly outperforms the multitone signals in terms of WPT. Moreover, it also demonstrates the effect of the transmission power $P_t$, as we observe that the harvested energy saturates after a certain $P_t$ is reached. The reason for this observation is attributed to the limited capability of a practical HPA at the transmitter. A typical HPA cannot transmit with any arbitrary amount of power and beyond a certain limit, it can only transmit with a fixed $P_t$, irrespective of the power of its input signal \cite{power}. In this figure, we observe this phenomenon around $P_t=25$ dBm and hence the saturating effect in the $P_t>25$ dBm range.

\section{Conclusion}

In this paper, we investigated the effects of conventional communication-based chaotic waveforms in WPT. In particular, we considered a point-to-point set-up, where a transmitter uses a DCSK generator and the receiver employs an analog correlator followed by an EH circuit. By taking into account the nonlinearities of the EH process, we proposed a novel WPT architecture and analyzed it in terms of the received signal's PAPR as well as the achieved harvested DC. We showed that both of these performance metrics are dependent on the parameters of the transmitted waveform. In addition, we demonstrated that DCSK waveforms outperform their unmodulated counterparts, in terms of EH performance, when a correlator is employed at the receiver. Based on this observation, a novel SR-DCSK-based waveform design has been investigated, which further enhances the WPT performance. Finally, we have shown that the proposed chaotic waveform outperforms the conventional multisine signals in terms of its WPT capability. An immediate extension of this work is to investigate  multi-dimensional chaotic signals for WPT, by considering a generalized frequency selective fading scenario.

\begin{figure}[!t]
\centering\includegraphics[width=\linewidth]{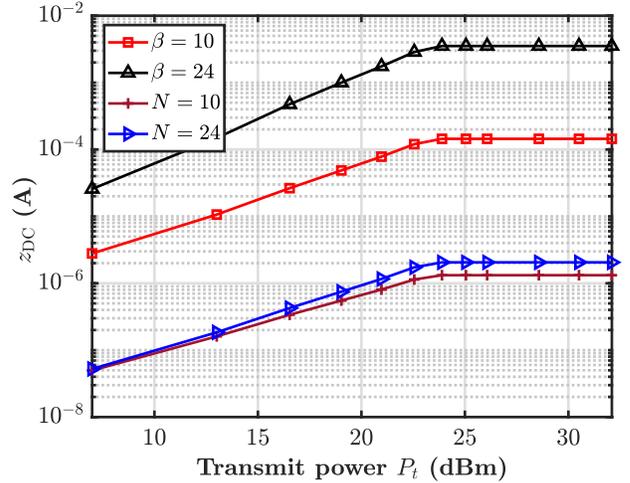}
\caption{Performance comparison of the proposed WPT-optimal SR-DCSK and multisine waveforms, considering HPA non-linearity; $m=4$.}
\label{fig:mult}
\end{figure}

\appendices

\section{Proof of Proposition \ref{prp1}}  \label{app0}

Without a correlator at the receiver, i.e. for $\psi=1$, the $\rm PAPR$ corresponding to the $l$-th transmitted symbol is
\begin{equation}
\mathrm{PAPR}=\frac{\max\limits_{l} \left\lbrace  \sum\limits_{k=2(l-1)\beta+1}^{2l\beta}|h_l|^2s_{l,k}^2 \right\rbrace }{\mathbb{E}\left\lbrace \sum\limits_{k=2(l-1)\beta+1}^{2l\beta}|h_l|^2s_{l,k}^2 \right\rbrace }.
\end{equation}
By considering a given channel instance $h_l$ and the invariant probability density function (PDF) of $x_k$, as \cite{ch2}
\begin{align}  \label{spdf}
f_X(x)=\begin{cases} 
\frac{1}{\pi\sqrt{1-x^2}}, & |x|< 1,\\
0, & \text{otherwise},
\end{cases}&
\end{align}
we obtain $\max\limits_{l} \left\lbrace  \sum\limits_{k=2(l-1)\beta+1}^{2l\beta}|h_l|^2s_{l,k}^2 \right\rbrace=2|h_l|^2\beta$ and
\begin{align} \label{deno}
&\mathbb{E}\left\lbrace \sum\limits_{k=2(l-1)\beta+1}^{2l\beta}|h_l|^2s_{l,k}^2 \right\rbrace& \nonumber \\
&=|h_l|^2\sum\limits_{k=2(l-1)\beta+1}^{2l\beta} \mathbb{E}\{s_{l,k}^2\} \nonumber \\
&=|h_l|^2\sum\limits_{k=2(l-1)\beta+1}^{2l\beta}\int_{-\infty}^{\infty}x^2f_X(x)dx=|h_l|^2\beta.
\end{align}
Hence, for $\psi=1$, we have
\begin{align}
{\rm PAPR} = \frac{2|h_l|^2\beta}{|h_l|^2\beta}=2.
\end{align}
On the other hand, for $\psi=2\beta$, i.e. with a correlator at the receiver,

\begin{equation}
\mathrm{PAPR}=\frac{\max\limits_{l}\left\lbrace \left( \sum\limits_{k=2(l-1)\beta+1}^{2l\beta}|h_l|s_{l,k} \right)^2 \right\rbrace}{\mathbb{E} \left\lbrace \left( \sum\limits_{k=2(l-1)\beta+1}^{2l\beta}|h_l|s_{l,k} \right)^2 \right\rbrace },
\end{equation}
where from (\ref{spdf}), we get
\begin{align}
\max\limits_{l}\left\lbrace \sum\limits_{k=2(l-1)\beta+1}^{2l\beta}|h_l|s_{l,k} \right\rbrace^2=4|h_l|^2\beta^2.
\end{align}
Note that for chaotic sequences generated by the Chebyshev map, we have $\mathbb{E}\left\lbrace s_{l,i}s_{l,j} \right\rbrace=0$ for $i \neq j$ \cite[Eq.~55]{nrml}. As such,

\begin{align}
\mathbb{E} \left\lbrace \left( \sum\limits_{k=2(l-1)\beta+1}^{2l\beta}|h_l|s_{l,k} \right)^2 \right\rbrace&=|h_l|^2\sum\limits_{k=2(l-1)\beta+1}^{2l\beta} \mathbb{E}\{s_{l,k}^2\} \nonumber \\
&=|h_l|^2\beta,
\end{align}
which follows from (\ref{deno}). Therefore, we have
\begin{align}
\mathrm{PAPR}=\frac{4|h_l|^2\beta^2}{|h_l|^2\beta}=4\beta.
\end{align}

\section{Proof of Theorem \ref{nakath1}}
\label{napp1}

As we consider a unit mean power Nakagami-$m$ distributed block fading channel,  the PDF of $|h|$ is
\begin{equation}    \label{nakadef}
f_{|h|}(z)=\frac{2m^mz^{2m-1}e^{-mz^2}}{\Gamma(m)}, \:\: \forall \:\:z \geq 0,
\end{equation}
where $\Gamma(\cdot)$ denotes the complete Gamma function and $m \geq 1$ controls the severity of the amplitude fading. The quantity $\sum\limits_{k=1}^{2\beta}s_k$ in \eqref{zdef1} can be alternatively written as
\begin{align} \label{nchrv}
\sum\limits_{k=1}^{2\beta}s_k &=(x_1+\dots+x_{\beta}+d(x_1+\dots+x_{\beta}))=p\sum\limits_{k=1}^{\beta}x_k,
\end{align}
where we have $p=1+d$. Then, by assuming equally likely transmissions of $d=\pm 1$, we have the PDF $f_P(p)=\frac{1}{2},$ $\forall$ $p\in \{0,2\}$. Hence, from \eqref{zdef1} we have
\begin{align}\label{zmc}
z_{\rm MC}&=\varepsilon_1\mathbb{E}\left\lbrace \left(  |h|\sum\limits_{k=1}^{2\beta}s_k \right) ^2 \right\rbrace  + \varepsilon_2\mathbb{E}\left\lbrace \left(  |h|\sum\limits_{k=1}^{2\beta}s_k \right) ^4 \right\rbrace \nonumber \\
&=\varepsilon_1\mathbb{E}\left\lbrace |h|^2 \right\rbrace \mathbb{E}\left\lbrace \left( p\sum\limits_{k=1}^{\beta}x_k \right) ^2 \right\rbrace \nonumber \\
& \qquad +\varepsilon_2\mathbb{E}\left\lbrace |h|^4 \right\rbrace \mathbb{E}\left\lbrace \left( p\sum\limits_{k=1}^{\beta}x_k \right) ^4 \right\rbrace.
\end{align}
Note that as we are considering a unit power Nakagami-$m$ fading scenario, we have $\mathbb{E}\left\lbrace |h|^2 \right\rbrace=1$ and
\begin{align}  \label{h4}
\mathbb{E}\left\lbrace |h|^4 \right\rbrace&=\frac{2m^m}{\Gamma(m)}\int_0^{\infty} z^{2m+3}e^{-mz^2} dz \nonumber \\
&=\frac{1}{m^2\Gamma(m)}\int_0^{\infty} v^{m+1}e^{-v}dv=\frac{(1+m)}{m},
\end{align}
which follows from the transformation $mz^2 \to v$. Then, the first term of \eqref{zmc} can be evaluated as
\begin{align}  \label{s2}
\mathbb{E}\left\lbrace \left( p\sum\limits_{k=1}^{\beta}x_k \right) ^2 \right\rbrace&=\mathbb{E}\left\lbrace p^2 \right\rbrace\mathbb{E}\left\lbrace \left( \sum\limits_{k=1}^{\beta}x_k \right) ^2 \right\rbrace \nonumber \\
&\overset{(a)}{=}2\mathbb{E}\left\lbrace \sum\limits_{k=1}^{\beta}x_k^2 + 2\sum_{\substack{l_1,l_2=1 \\ l_1\neq l_2}}^{\beta}x_{l_1}x_{l_2} \right\rbrace \nonumber \\
&=2\sum\limits_{k=1}^{\beta}\mathbb{E}\left\lbrace x_k^2 \right\rbrace+4\sum_{\substack{l_1,l_2=1 \\ l_1\neq l_2}}^{\beta}\mathbb{E}\left\lbrace x_{l_1}x_{l_2} \right\rbrace \nonumber \\
&=\beta,
\end{align}
where (a) follows from $\mathbb{E}\{p^2\} = 2$ and the multinomial theorem, and the final result follows from
\begin{equation}  \label{sq}
\mathbb{E}\left\lbrace x_k^2 \right\rbrace=\int_{-1}^{1}\frac{x_k^2dx}{\pi\sqrt{1-x_k^2}}=\frac{1}{2}
\end{equation}
and by using \cite[Eq.~55]{nrml}. Similarly, by using the multinomial theorem, the second term of \eqref{zmc} can be expanded as
\begin{align}  \label{s4}
&\mathbb{E}\left\lbrace \left( p\sum\limits_{k=1}^{\beta}x_k \right) ^4 \right\rbrace \nonumber \\
&=\mathbb{E}\left\lbrace p^4 \right\rbrace\mathbb{E}\left\lbrace \left( \sum\limits_{k=1}^{\beta}x_k \right) ^4 \right\rbrace \nonumber \\
&\overset{(b)}{=}\mathbb{E}\left\lbrace p^4 \right\rbrace\mathbb{E}\left\lbrace \sum\limits_{k_1+k_2+\cdots+k_{\beta}=4} \frac{4!}{k_1! \: k_2! \: \cdots \: k_{\beta}!} \prod\limits_{i=1}^{\beta} x_i^{k_i} \right\rbrace \nonumber \\
&\overset{(c)}{=}8 \left( \sum\limits_{k=1}^{\beta}\mathbb{E}\left\lbrace x_k^4 \right\rbrace+ 3\mathbb{E}^2\left\lbrace x_k^2 \right\rbrace\beta(\beta-1) \right)   \nonumber \\
&\overset{(d)}{=}8 \left( \frac{3\beta}{8}+ \frac{3}{4}\beta (\beta-1) \right)=3\beta(2\beta-1),
\end{align}
where $(b)$ follows from the multinomial theorem. Furthermore, $(c)$  follows from $\mathbb{E}\{p^4\} = 0\times\frac{1}{2}+2^4\times\frac{1}{2}=8$,
\begin{equation}
\mathbb{E}\left\lbrace x_k \right\rbrace=\int_{-1}^{1}\frac{x_kdx}{\pi\sqrt{1-x_k^2}}=0,
\end{equation}
and $(d)$ follows from \eqref{sq} and
\begin{equation}  \label{x4}
\mathbb{E}\left\lbrace x_k^4 \right\rbrace=\int_{-1}^{1}\frac{x_k^4dx}{\pi\sqrt{1-x_k^2}}=\frac{3}{8}.
\end{equation}
By combining \eqref{zmc}, \eqref{h4}, \eqref{s2}, and \eqref{s4}, we obtain
\begin{equation}
z_{\rm MC}=\varepsilon_1\beta+\varepsilon_2\frac{3(1+m)}{m}\beta(2\beta-1).
\end{equation}

\section{Proof of Theorem \ref{nakath2}} \label{napp2}

From the expression \eqref{zdef2}, we have
\begin{align}
z_{\rm MNC}&=\varepsilon_1\mathbb{E}\left\lbrace |h|^2\sum\limits_{k=1}^{2\beta}s_k^2 \right\rbrace  + \varepsilon_2\mathbb{E}\left\lbrace |h|^4\sum\limits_{k=1}^{2\beta}s_k^4 \right\rbrace \nonumber \\
&=\varepsilon_1\mathbb{E}\left\lbrace |h|^2 \right\rbrace \mathbb{E}\left\lbrace \sum\limits_{k=1}^{2\beta}s_k^2 \right\rbrace+\varepsilon_2\mathbb{E}\left\lbrace |h|^4 \right\rbrace \mathbb{E}\left\lbrace \sum\limits_{k=1}^{2\beta}s_k^4 \right\rbrace \nonumber \\
&=\varepsilon_1\left(1+\sum\limits_{k=1}^{\beta}\mathbb{E}\left\lbrace d_k^2 \right\rbrace \right) \left( \sum\limits_{k=1}^{\beta} \mathbb{E}\left\lbrace x_k^2 \right\rbrace \right) \nonumber \\
&\qquad+\varepsilon_2\left(1+\sum\limits_{k=1}^{\beta}\mathbb{E}\left\lbrace d_k^4 \right\rbrace \right) \left( \sum\limits_{k=1}^{\beta} \mathbb{E}\left\lbrace x_k^4 \right\rbrace \right) \nonumber \\
&\overset{(a)}{=}\varepsilon_1\beta+\varepsilon_2\frac{3(1+m)}{4m}\beta,
\end{align}
where $(a)$ follows from \eqref{sq}, \eqref{x4} and by assuming equally likely transmissions of $d_k=\pm 1,$ $\forall$ $k$.

\section{Proof of Proposition \ref{prop2}} \label{napp3}

We prove the proposition in the form of two separate sub-cases: i) with, and ii) without the correlator.

\subsubsection{With correlator}

Based on \eqref{zdef1} and \eqref{corr3}, we obtain the harvested DC as

\begin{equation}
z_{\rm UM,C}=\varepsilon_1\mathbb{E}\left\lbrace \left(  |h|\sum\limits_{k=1}^{2\beta}x_k \right) ^2 \right\rbrace+ \varepsilon_2\mathbb{E}\left\lbrace \left(  |h|\sum\limits_{k=1}^{2\beta}x_k \right) ^4 \right\rbrace.
\end{equation}
Hence,
\begin{align}
z_{\rm UM,C}&=\varepsilon_1\mathbb{E}\left\lbrace |h|^2 \right\rbrace\mathbb{E}\left\lbrace \left( \sum\limits_{k=1}^{2\beta}x_k \right) ^2 \right\rbrace \nonumber \\
& \qquad +\varepsilon_2\mathbb{E}\left\lbrace |h|^4 \right\rbrace\mathbb{E}\left\lbrace \left( \sum\limits_{k=1}^{2\beta}x_k \right) ^4 \right\rbrace \nonumber \\
&=\varepsilon_1\mathbb{E}\left\lbrace \sum\limits_{k=1}^{2\beta}x_k^2 + 2\sum_{\substack{l_1,l_2=1 \\ l_1\neq l_2}}^{2\beta}x_{l_1}x_{l_2} \right\rbrace \nonumber \\
&\qquad+\varepsilon_2\frac{(1+m)}{m}\mathbb{E}\left\lbrace \sum\limits_{k=1}^{2\beta}x_k^4 + 6\sum_{\substack{l_1,l_2=1 \\ l_1\neq l_2}}^{\beta(2\beta-1)}x_{l_1}^2x_{l_2}^2 \right\rbrace \nonumber \\
&=\varepsilon_1\sum\limits_{k=1}^{2\beta} \mathbb{E} \left\lbrace x_k^2 \right\rbrace \nonumber \\
& + \varepsilon_2\frac{(1+m)}{m} \left( \sum\limits_{k=1}^{2\beta} \mathbb{E} \left\lbrace x_k^4 \right\rbrace + 6\beta(2\beta-1)\mathbb{E}^2\left\lbrace x_k^2 \right\rbrace \right) \nonumber \\
&\overset{(a)}{=}\varepsilon_1\beta+\varepsilon_2\frac{3(1+m)}{m}\beta(\beta-\frac{1}{4}),
\end{align}
where $(a)$ follows from \eqref{sq} and \eqref{x4}.

\subsubsection{Without correlator}

Based on \eqref{zdef2} and \eqref{corr3}, we have
\begin{align}
z_{\rm UM,NC}&=\varepsilon_1\mathbb{E}\left\lbrace |h|^2\sum\limits_{k=1}^{2\beta}x_k^2 \right\rbrace  + \varepsilon_2\mathbb{E}\left\lbrace |h|^4\sum\limits_{k=1}^{2\beta}x_k^4 \right\rbrace \nonumber \\
&\overset{(b)}{=}\varepsilon_1\times 2\beta \times \frac{1}{2} +\varepsilon_2\times \frac{(1+m)}{m}\times 2\beta \times\frac{3}{8} \nonumber \\
&=\varepsilon_1\beta+\varepsilon_2\frac{3(1+m)}{4m}\beta,
\end{align}
where $(b)$ follows from \eqref{sq} and \eqref{x4}.

\section{Proof of Theorem \ref{nakath3}}  \label{napp4}

We prove the theorem in the form of two separate sub-cases with $\beta_r=0$ and $\beta_r>0.$

\subsection{Case $\beta_r=0$}

In this case, we obtain the harvested energy $z_{\rm SR}$ as
\begin{align}
z_{\rm SR}&=\varepsilon_1\mathbb{E}\left\lbrace \left( d \sum\limits_{k=1}^{\beta}x_k \right) ^2 \right\rbrace  \nonumber \\
& \qquad + \varepsilon_2\frac{(1+m)}{m}\mathbb{E}\left\lbrace \left( d\sum\limits_{k=1}^{\beta}x_k \right) ^4 \right\rbrace.
\end{align}
From Appendix \ref{napp1}, we assume an equally likely transmission of $d=\pm 1$ and by following a similar framework, we obtain
\begin{align}
z_{\rm SR}&=\varepsilon_1\beta^2\mathbb{E}\left\lbrace d^2 \right\rbrace \mathbb{E} \left\lbrace x_k^2 \right\rbrace + \varepsilon_2\beta^4\mathbb{E}\left\lbrace d^4 \right\rbrace \mathbb{E} \left\lbrace x_k^4 \right\rbrace \nonumber \\
&=\frac{1}{2}\varepsilon_1\beta^2+\frac{3(1+m)}{8m}\varepsilon_2\beta^4.
\end{align}

\subsection{Case $\beta_r>0$}

Based on \eqref{zdef1} and \eqref{symsr}, we obtain
\begin{align}  \label{lappp}
z_{\rm SR}&=\varepsilon_1\mathbb{E}\left\lbrace \left(  \sum\limits_{k=1}^{\beta+\beta_r}s_k \right) ^2 \right\rbrace \nonumber \\
& \qquad + \varepsilon_2\frac{(1+m)}{m}\mathbb{E}\left\lbrace \left( \sum\limits_{k=1}^{\beta+\beta_r}s_k \right) ^4 \right\rbrace,
\end{align}
where from \eqref{symsr}, we have
\begin{equation}
\sum\limits_{k=1}^{\beta+\beta_r}s_k=(1+\zeta d)\sum\limits_{k=1}^{\beta_r}x_k.
\end{equation}
Then, the first term of \eqref{lappp} can be evaluated as
\begin{align}
\mathbb{E}\left\lbrace \left(  \sum\limits_{k=1}^{\beta+\beta_r}s_k \right) ^2 \right\rbrace&=\mathbb{E}\left\lbrace \left(  1+\zeta d \right) ^2 \right\rbrace \mathbb{E} \left\lbrace \left(  \sum\limits_{k=1}^{\beta_r}x_k \right) ^2 \right\rbrace,
\end{align}
where
\begin{equation}  \label{dl1}
\mathbb{E}\left\lbrace \left(  1+\zeta d \right) ^2 \right\rbrace=\frac{1}{2}\left\lbrace \left(  1+\zeta \right)^2+\left(  1-\zeta \right) ^2 \right\rbrace=1+\zeta^2,
\end{equation}
and
\begin{equation}  \label{dl2}
\mathbb{E} \left\lbrace \left(  \sum\limits_{k=1}^{\beta_r}x_k \right) ^2 \right\rbrace=\sum\limits_{k=1}^{\beta_r}\mathbb{E} \left\lbrace x_k^2 \right\rbrace=\frac{\beta_r}{2}.
\end{equation}
By combining \eqref{dl1} and \eqref{dl2}, we obtain
\begin{equation}  \label{lf}
\mathbb{E}\left\lbrace \left(  \sum\limits_{k=1}^{\beta+\beta_r}s_k \right) ^2 \right\rbrace=\frac{\beta_r}{2} \left( 1+\zeta^2 \right)=\frac{\beta_r^2+\beta^2}{2\beta_r}.
\end{equation}
Similarly, by using the multinomial theorem, the second term in \eqref{lappp} can be expanded as
\begin{align}
\mathbb{E}\left\lbrace \left(  \sum\limits_{k=1}^{\beta+\beta_r}s_k \right) ^4 \right\rbrace&=\mathbb{E}\left\lbrace \left(  1+\zeta d \right) ^4 \right\rbrace \mathbb{E} \left\lbrace \left(  \sum\limits_{k=1}^{\beta_r}x_k \right) ^4 \right\rbrace,
\end{align}
where
\begin{align}  \label{dnl1}
\mathbb{E}\left\lbrace \left(  1+\zeta d \right) ^4 \right\rbrace&=\frac{1}{2}\left\lbrace \left(  1+\zeta \right)^4+\left(  1-\zeta \right) ^4 \right\rbrace \nonumber \\
&=1+6\zeta^2+\zeta^4
\end{align}
and
\begin{align}  \label{dnl2}
\mathbb{E} \left\lbrace \left(  \sum\limits_{k=1}^{\beta_r}x_k \right) ^4 \right\rbrace&=\sum\limits_{k=1}^{\beta_r}\mathbb{E}\left\lbrace x_k^4 \right\rbrace+ 3\mathbb{E}^2\left\lbrace x_k^2 \right\rbrace\beta(\beta-1) \nonumber \\
&=\frac{3\beta_r}{8}+ \frac{3}{4}\beta_r (\beta_r-1).
\end{align}
By combining \eqref{dnl1} and \eqref{dnl2}, we obtain
\begin{equation}  \label{nlf}
\mathbb{E}\left\lbrace \left(  \sum\limits_{k=1}^{\beta+\beta_r}s_k \right) ^4 \right\rbrace=\frac{3}{8} \left( 1+6\zeta^2+\zeta^4 \right) \left( 2\beta_r^2-\beta_r \right).
\end{equation}
Hence, we have $z_{\rm SR}$ equal to

\begin{align}
z_{\rm SR}\!&=\!\varepsilon_1\mathbb{E}\left\lbrace \left(  \sum\limits_{k=1}^{\beta+\beta_r}s_k \right) ^2 \right\rbrace  + \varepsilon_2\frac{(1+m)}{m}\mathbb{E}\left\lbrace \left( \sum\limits_{k=1}^{\beta+\beta_r}s_k \right) ^4 \right\rbrace \nonumber \\
&=\varepsilon_1 \frac{\beta_r^2+\beta^2}{2\beta_r}+\varepsilon_2\frac{3(1+m)}{8m}\left( 1+6\zeta^2+\zeta^4 \right) \left( 2\beta_r^2-\beta_r \right).
\end{align}

\bibliographystyle{IEEEtran}
\bibliography{refs}
\end{document}